\documentclass{article}
\usepackage{ijcai22}

\usepackage{times}
\usepackage{soul}
\usepackage{url}
\usepackage[hidelinks]{hyperref}
\usepackage[utf8]{inputenc}
\usepackage[small]{caption}
\usepackage{graphicx}
\usepackage{amsmath}
\usepackage{amsthm}
\usepackage{booktabs}
\usepackage{algorithm}
\usepackage{algorithmic}
\usepackage{tcolorbox}  
\newcommand{\set}[1]{\left\{ #1 \right\}}
\newcommand{\brackets}[1]{\left(#1\right)}

\newcommand{\calO}{\mathcal{O}}
\newcommand{\calD}{\mathcal{D}}
\newcommand{\corr}{\mathit{corr}}
\newcommand{\poss}{\mathit{poss}}
\newcommand{\vect}{\mathbf{t}}

\usepackage{color}
\usepackage{amsmath}
\usepackage{amssymb}
\usepackage{nicefrac}
\usepackage{enumerate}
\usepackage{amsthm}
\usepackage{booktabs}
\usepackage{algorithm}
\usepackage{algorithmic} 
\newtheorem{example}{Example}

\newtheorem{proposition}{Proposition}
\newtheorem{observation}{Observation}
\newtheorem{theorem}{Theorem}

\usepackage{cleveref}
\usepackage{todonotes}

\title{Better Collective Decisions via Uncertainty Reduction}

\author{
Shiri Alouf-Heffetz$^{1,4}$
\and
Laurent Bulteau$^2$\and
Edith Elkind$^3$\and
Nimrod Talmon$^1$\And
Nicholas Teh$^3$
\affiliations
$^1$Ben-Gurion University, Israel\\
$^2$LIGM, CNRS, Universit\'e Gustave Eiffel, France\\
$^3$Department of Computer Science, University of Oxford, UK\\
$^4$TwoFiveOne Research
\emails
shirihe@post.bgu.ac.il,
laurent.bulteau@univ-eiffel.fr,
elkind@cs.ox.ac.uk,
talmonn@bgu.ac.il,
nicholas.teh@cs.ox.ac.uk
}

\begin{document}

\maketitle

\begin{abstract}
    We consider an agent community wishing to decide on several binary issues by means of issue-by-issue majority voting. For each issue and each agent, one of the two options is better than the other. However, some of the agents may be confused about some of the issues, in which case they may vote for the option that is objectively worse for them. A benevolent external party wants to help the agents to 
    select the majority-preferred option for as many issues as possible. This party may have one of the following tools at its disposal:
    (1) educating some of the agents, so as to enable them to vote correctly on all issues, 
    (2) appointing a subset of highly competent agents to make decisions on behalf of the entire group, or 
    (3) guiding the agents on how to delegate their votes to other agents, in a way that is consistent with the agents' opinions. 
    For each of these tools, we study the complexity of the decision problem faced by this external party, obtaining both NP-hardness and fixed-parameter tractability results.

\end{abstract}

\section{Introduction}
A rural community needs to decide whether to install wind turbines on a nearby hill.
Some 40\% of the residents are certain that this is a good idea, because they are committed to renewable energy. About 30\% of the residents are firmly against the proposal, because of the construction noise or the impact on the views from their homes. However, the remaining 30\% of the residents struggle to make up their minds: they are generally in favor of wind power, but worry that the location chosen for the project may be on the migration route of a rare bird species. A local environmental charity wants to help this community to make a decision that is consistent with the majority's preferences. By consulting with scientists, the charity concludes that 
the selected location is unlikely to present a danger for the birds, so building the turbines is the `correct' decision. As the decision will be made by a majority vote at a community meeting, the charity needs to reach out to (some of) the undecided voters and explain that it is in their best interest to vote in favor of the wind turbines. Another approach would be to discourage the confused residents from participating in the meeting, or to suggest to them that they may want to delegate their votes to more competent voters; these strategies may result in decisions that are better for the entire community, including the confused voters themselves.

Inspired by this scenario, in this work we study the problem faced by a benevolent party that wants to help a group of agents to make majority-preferred decisions on a number of binary issues. For each agent and each issue, one of the two options is `correct', in the sense that this is the option the agent would have selected if they could invest time and effort into studying it. However, just as in our example, some of the agents may be uncertain about some of the issues, in which case they may vote against their best interest. Examples of such confusion abound in the real world and are well-documented, e.g., in the context of Brexit or nuclear power.

In our setting, there are multiple issues and the benevolent party does not have its own preferences over the outcomes. Rather, it wants to maximize the number of issues on which the group is guaranteed to make the decision that matches the true preferences of the majority. This party, which may be a charity, an impartial governmental organization, or an ad-hoc working group, may have several tools at its disposal. For instance, it may be able to reach out to a subset of the agents and offer them an opportunity to learn more about the issues (e.g., by presenting information in an accessible manner, or holding a Q\&A session with an expert). Alternatively, it may discourage some of the agents from participating in the vote.

The benevolent party may also help the agents with delegation decisions. Specifically, we assume that the voting mechanism used by the community supports delegation: an agent may delegate her vote to another agent, who will vote on her behalf on all issues. The agents are assumed to be willing to delegate to other agents who have similar preferences (i.e., $i$ would not delegate to $j$ if they disagree on an issue they are both certain about), but are more knowledgeable.
The benevolent party can suggest delegation options to some of the agents, in a way that respects these constraints and maximizes the number of good decisions.

\subsection{Our Contribution}
We develop a formal model that enables us to reason about voters' confusion and ways to mitigate it. In this paper, we focus on independent binary issues, but our ideas can be extended to more complex decision-making scenarios. Our approach is worst-case: we want to maximize the number 
of issues that are `safe', in the sense that the majority-supported outcome is guaranteed no matter how the uncertain voters cast their ballots.
We consider both a general discrete issue space and a one-dimensional setting, where the 
issues are ordered in such a way that, for each agent, the set
of issues on which she would benefit from a positive decision forms an interval
of this ordering, and she is only uncertain about issues that are close to the endpoints of her `positive interval'. For instance, consider a vote over tax rates, and a voter whose policy preferences are consistent with a tax rate in the [25\%, 35\%] range: this voter may confidently reject a proposal to set the tax rate to 15\% or 45\% and support a proposal to set it to 30\%, but may struggle to make up her mind about the proposals in the [23\%, 27\%] range or the [34\%, 38\%] range.

We formulate three computational problems that model, respectively, educating the agents, preventing
some agents from participating, and helping with delegation decisions.
We establish that these problems are NP-hard even in the one-dimensional setting, 
but show that they are fixed-parameter tractable---even in the general setting---both with respect to the number of voters and the number of issues. Moreover, we consider a natural special case of the one-dimensional setting in which all three problems are polynomial-time solvable.
We omit some proofs due to space constraints.

\subsection{Related Work}
%
The analysis of voting outcomes under uncertainty is a prominent
topic in computational social choice (see, e.g., \cite{walsh2007uncertainty,hazon08,bachrach10,wojtas12,kenig19}; 
however, it is usually assumed that it is an external party, rather than the voters themselves, who is uncertain about the votes.
There is a large body of work in political science that considers voter education (see, e.g., \cite{lupia98,bullock11,boudreau19} and the references therein);
however, it does not engage with the associated algorithmic issues.
The computational problem associated with removing voters is closely related to election control by deleting voters (see, e.g., the survey of Faliszewski and Rothe~\shortcite{faliszewski2016control}); however, in the control literature, it is assumed that the party engaging in voter deletion pursues its own goals rather than tries to implement the popular opinion.
Finally, in the context of vote delegation, we mention works on liquid democracy from the perspective of discovering a good outcome \cite{brill2018pairwise,bloem19,kahng2021liquid,golz21}
in particular, Cohensius \textit{et al.}~\shortcite{cohensius2016proxy} and Green-Armytage~\shortcite{green2015direct} consider vote delegation in the one-dimensional Euclidean domain.


\section{Preliminaries}\label{sec:prelim}

We will now define our model and formulate the computational problems that we are going to study.

\paragraph{Voting instances}
There is a {\em proposal space} $P$ and a set of $n$ {\em agents} $N$.
For each agent $i\in N$, the space $P$ is split as $P=P^+(i)\cup P^-(i)$ with $P^+(i)\cap P^-(i)=\varnothing$; the proposals in $P^+(i)$ are {\em beneficial} for $i$, so $i$ would want to approve them, while proposals in $P^-(i)$ are not beneficial for $i$.
Also, for each agent $i\in N$, the space $P$ is split as $P=P^!(i)\cup P^?(i)$ with $P^!(i)\cap P^?(i)=\varnothing$; agent
$i$ is {\em certain} about the proposals in $P^!(i)$ and \emph{uncertain} about the proposals in $P^?(i)$. 
Thus, an agent $i$ can be described by a tuple 
$\pi(i) = (P^+(i), P^!(i))\in P\times P$, which we will call agent $i$'s {\em belief}; we write $\Pi=(\pi(1), \dots, \pi(n))$ and refer to the triple $I = (P, N, \Pi)$ as a {\em voting instance}.
For $i\in N$ and $*\in\{+, -\}$, $\dagger\in\{!, ?\}$, let $P^{*\dagger}(i)=P^*(i)\cap P^\dagger(i)$. 
When casting an approval ballot, agent~$i$ will vote for all proposals in $P^{+!}(i)$ and against all proposals in $P^{-!}(i)$. Her vote over proposals in $P^?(i)$, however, may be arbitrary; in particular, she may disapprove proposals in $P^{+?}(i)$ and approve proposals in $P^{-?}(i)$.

\paragraph{Correct and possible outcomes}
For each $p\in P$ and 
$\star\in\{+, -, ?, !\}$, let
$N^\star(p) = \{i\in N: p\in P^\star(i)\}$.
We say that a proposal $p$ is {\em good} if
$|N^+(p)| \ge |N|/2$ and {\em bad} if
$|N^-(p)| \ge |N|/2$. Note that a proposal
can be both good and bad: this happens if it is beneficial for exactly $|N|/2$ agents.
An outcome $z\in\{0, 1\}$ is {\em correct for $p$} if $p$ is good and $z=1$ (`approve'), or if $p$ is bad and $z=0$ (`disapprove').
We denote the set of correct outcomes for $p$ by $\corr(p, I)$.
Because of agent uncertainty, 
the outcome of an approval vote 
is not always correct:
we say that $1$ (respectively, $0$) is a {\em possible outcome} for $p$ if 
$|N^+(p)\cup N^?(p)|\ge |N|/2$
(respectively, if $|N^-(p)\cup N^?(p)|\ge |N|/2$).
Note that $\corr(p, I)\subseteq \poss(p, I)$ for each $p\in P$, 
but for some $p\in P$ we may have 
$\poss(p, I)\setminus\corr(p, I)\neq\varnothing$.

\paragraph{Safe proposals and safe zones}
Given a voting instance $I=(P, N, \Pi)$, we say that a proposal $p\in P$ is \emph{safe} if $\poss(p, I)=\corr(p, I)$ and \emph{unsafe} otherwise.
The set of safe proposals 
is called the \emph{safe zone} for $I$.

\begin{example}\label{ex:3agents}
Consider an instance $I$ 
with $N = \{1, 2, 3\}$, $P=\{p_1, p_2, p_3, p_4, p_5\}$, 
and $\Pi$ given by
$P^+(1)=\{p_2, p_3, p_4\}$, $P^?(1)=\{p_2, p_4\}$,
$P^+(2)=\{p_2, p_3\}$,  $P^?(2)=\{p_1\}$,
$P^+(3)=\{p_1, p_2, p_3\}$,
and $p^?(3)=\varnothing$.
%

The first agent benefits from $p_3$ and knows this, 
so she is going to approve $p_3$. She also benefits
from $p_2$, but is uncertain about this proposal, so she may disapprove it. 
We have $\corr(p_1, I)=\{0\}$, since agents $1$ and $2$ do not benefit from $p_1$ (so $N^-(p_1)=\{1, 2\}$, $N^+(p_1)=\{3\}$). However, $\poss(p_1,I) = \{0, 1\}$, 
because agent~2 is uncertain about $p_2$ (i.e., $p_2\in P^?(2)$) and may vote either way on it (i.e., $|N^+(p_1)\cup N^?(p_1)|=2$, $|N^-(p_1)\cup N^?(p_1)|=2$). Hence, $p_1$ is not safe. The other four proposals, however, are safe.
\end{example}

\subsubsection{Proposal Spaces}
In general, our framework allows for infinite proposal spaces.
However, in this work we will focus on settings where the proposal space is \emph{finite}, i.e., $P=\{p_1, \dots, p_m\}$.
A natural restricted case of our model is a one-dimensional setting, 
where $P$ is ordered as $p_1\prec\dots\prec p_m$, and agent preferences respect this order, as described below. 
We say that a set of proposals 
$P'\subseteq P$ forms an {\em interval of $\prec$} if 
$P'=\varnothing$ or $P'=\{p_\ell: j\le \ell\le k\}$ for some $1\le j\le k\le m$. 
In the one-dimensional setting we assume that, for each agent $i\in N$, the set $P^+(i)$ forms an interval of $\prec$.
This approach takes inspiration from the notion of interval approval domains~\cite{approvaldomains}.

Furthermore, in the one-dimensional case we expect an agent to be certain that she dislikes extreme proposals and likes proposals in the center of her interval; however, she may be uncertain about proposals that are close to the endpoints of her interval. Hence, we assume that, for each $i\in N$ with $P^+(i)=\{p_j, \dots, p_k\}$, we have $P^?(i)=P^?_L(i)\cup P^?_R(i)$, where both $P^?_L(i)$ and $P^?_R(i)$ are intervals of $\prec$ that satisfy the following conditions:
(1) $P^?_L(i)=\varnothing$ or $P^?_L(i)\cap\{p_{j-1}, p_j\}\neq\varnothing$; 
(2) $P^?_R(i)=\varnothing$ or $P^?_R(i)\cap\{p_{k}, p_{k+1}\}\neq\varnothing$.

\begin{example}\label{ex:1D}
Note that the instance $I$ from Example~\ref{ex:3agents} is one-dimensional.
However, if we were to add either of the two agents $4, 5$
with $P^+(4) = \{p_1, p_5\}$, $P^?(4)=\varnothing$, 
$P^+(5)=\{p_2, p_3, p_4\}$, $P^?(5)=\{p_3\}$,
this would no longer be the case:
$P^+(4)$ is not an interval, and $P^?(5)$ fails conditions (1) and~(2).
\end{example}

As the one-dimensional setting is essentially a domain restriction, algorithms for general finite proposal spaces apply in the one-dimensional setting, while NP-hardness results for the one-dimensional setting imply hardness for general spaces.
Therefore, in what follows, for all problems we consider, 
we prove hardness results for the one-dimensional case (Section~\ref{sec:NP}) and develop FPT 
algorithms for the general case (Section~\ref{sec:FPT}).
Additionally, in Section~\ref{sec:rad} we present
polynomial-time algorithms for a special, more restricted case of 
the one-dimensional model.

\subsubsection{Algorithmic Challenges}
%
%
We consider three approaches to eliminating uncertainty:
  (1) \textit{educating} agents,
  (2) \textit{removing} agents, and 
  (3) \textit{delegating} votes (known as {\em liquid democracy}).
In each case, the goal is to maximize the size of the safe zone; an important special case of this task is to ensure that all projects are safe.

\paragraph{Educating agents}
In the first approach we consider, at a unit cost, we can educate one agent so as to entirely remove her uncertainty: if $i$ has been educated, then the set $P^?(i)$ becomes empty, so, when voting, agent $i$ will approve proposals in $P^+(i)$ and disapprove proposals in $P^-(i)$.

\begin{tcolorbox}
\textsc{Educating Agents Problem (EAP)}:\\
\textbf{Input}: A voting instance $I = (P, N, \Pi)$, a budget $\kappa$, and a parameter $\lambda\in \mathbb N$.\\
\textbf{Question}: Can we educate at most $\kappa$ agents in $N$ so that in the resulting instance $I'$ the number of safe proposals is at least $\lambda$?
\end{tcolorbox}
\noindent Note that the input to EAP contains the list $\Pi$; that is, we assume that, when deciding which agents to educate, we know which proposals are beneficial for them, i.e., we know the sets $P^+(i)$, $P^-(i)$ for all $i\in N$. Of course, this is not always the case; in particular, the agents themselves can be assumed to know the sets $P^{+!}(i)$, $P^{-!}(i)$ and $P^?(i)$, but not $P^+(i)$ and $P^-(i)$. However, we expect that in practice we can usually make good predictions based on, e.g., demographic criteria.

\paragraph{Denying access}
The next approach we consider is to prevent up to $\kappa$ agents
from participating in the vote (equivalently, this approach can be viewed as selecting a committee). However, we still want to make decisions in a way that is beneficial to the majority of all agents, including those who are not invited to vote. That is, we consider the set of correct outcomes for each proposal with respect to the original voting instance $I$ and the set of possible outcomes with respect to the modified instance $I'$, with up to $\kappa$ agents removed.
Note that it may be the case that $\corr(p, I)=\{0, 1\}$, 
but $\poss(p, I')$ is a singleton; we consider such outcomes as acceptable.

\begin{tcolorbox}
\textsc{Denying Access Problem (DAP)}:\\
\textbf{Input}: A voting instance $I = (P, N, \Pi)$, a budget $\kappa$, and a parameter $\lambda\in \mathbb N$.\\
\textbf{Question}: Can we remove at most $\kappa$ agents from $N$ 
to obtain an instance $I'$ such that the number of proposals 
$p\in P$ with $\corr(p, I)\supset\poss(p, I')$ is at least~$\lambda$?
\end{tcolorbox}
\noindent Note that, to solve DAP, we do not need full access to the 
`hidden' sets $P^+(i), P^+(i)$; it suffices to have access to the `known' sets $P^{+!}(i)$, $P^{-!}(i)$ and $P^?(i)$ for each $i\in N$ as well as the correct outcome for each proposal (which we expect to be easier to estimate than individual preferences).

\paragraph{Allowing delegations}
The last approach we consider is to allow the agents to delegate their votes to other agents, transitively. An agent that accumulates $t$ votes participates in the election with voting weight $t$.
This framework is known as {\em liquid democracy}~\cite{green2015direct}.
We assume that our agents are rather conservative in their delegation decisions: we say that agent $i$ is {\em willing to delegate her vote to agent $j$} if 
(1) $P^{+!}(i)\cap P^{-!}(j)=\varnothing$ and
$P^{-!}(i)\cap P^{+!}(j)=\varnothing$, and 
(2) $|P^!(i)| < |P^!(j)|$. Condition~(1) indicates
that agent $i$ cannot observe any disagreement between herself and agent $j$, and condition~(2) indicates that agent $j$ is more knowledgeable than agent $i$.
A {\em delegation graph consistent with an instance $I$} is a directed acyclic graph $\calD$ over the vertex set $N$ such that the outdegree of each vertex is at most $1$; it may contain an arc $(i, j)$ only if agent $i$ is willing to delegate to agent $j$. Given a graph $\calD$, for each $i\in N$ we denote the unique sink reachable from $i$ by $s(i)$; the agent $s(i)$ is the `guru' of agent $i$, so that $i$ will adopt the beliefs of $s(i)$.
The graph $\calD$ corresponds to a modified instance $I'(\calD)=(N, P, \Pi')$ 
such that $\pi'(i)=\pi(s(i))$ for each agent $i\in N$; intuitively, in $I'(\calD)$ agent $i$ copies the ballot of her guru $s(i)$. We want to construct a delegation graph that results in decisions that benefit the majority of the agents for as many proposals as possible.

\begin{tcolorbox}
\textsc{Constrained Liquid Democracy (CLD)}:\\
\textbf{Input}: A voting instance $I = (P, N, \Pi)$, and a parameter $\lambda\in \mathbb N$.\\
\textbf{Question}: Is there a delegation graph $\calD$ consistent with $I$ such that the number of proposals 
$p\in P$ with $\corr(p, I)\supset\poss(p, I'(\calD))$ is at least $\lambda$?
\end{tcolorbox}
\noindent Importantly, in a way, our model implicitly assumes that the delegation decisions are made in a centralized manner; however, no agent can be asked to delegate their vote in a way that is unacceptable to them.

\begin{example}
Consider again the instance $I$ of Example~\ref{ex:3agents}.
Note that, in $I$, we can make all proposals safe by educating agent~$2$.
If we can only remove agents, 
then the only way to ensure the correct outcome on $p_1$ 
is to remove agents $2$ and $3$; however, if we do so, 
we may get an incorrect outcome on $p_2$ and $p_4$, 
so no removal strategy can ensure the correct outcome on all proposals.

Finally, observe that agent~$1$ can delegate to agent~$2$ (they agree 
on the two proposals on which they are both certain, i.e., 
$p_3$ and $p_5$, and $|P^!(2)|>|P^!(1)|$), and 
agent~$2$ can delegate to agent~$3$, 
but agent~$1$ cannot directly delegate to agent~$3$.
However, if $1$ delegates to $2$ and $2$ delegates to $3$, then agent~$3$ becomes the guru for agent~$1$. No delegation strategy ensures a correct outcome on all proposals.
\end{example}

\section{Intractability Results}\label{sec:NP}

It follows from our definitions that checking whether a given proposal
is safe or determining the number of safe proposals can be done in polynomial time.
Next we show that EAP, DAP, and CLD are NP-complete.
Containment in NP follows immediately;
thus, in what follows, we focus on showing that 
these problems are NP-hard.
Perhaps surprisingly, our hardness results hold even for the one-dimensional setting. 
\begin{theorem}\label{thm:1D-EAP_is_hard}
  {\sc EAP} is {\em NP}-complete, even in the one-di\-men\-si\-onal 
  proposal space.
\end{theorem}
\begin{proof}
    We provide a reduction from {\sc Vertex Cover on Cubic Graphs (VC-CG)}. An instance of this problem comprises of  an undirected cubic graph $G=(V, E)$ (i.e., the degree of each vertex $v\in V$ is exactly $3$) and an integer $\kappa$; it is a yes-instance if $G$ admits a vertex cover of size~$\kappa$ (i.e., a subset $V'\subseteq V$ that contains at least one endpoint of each edge), and a no-instance otherwise. This problem is known to be NP-complete~\cite{GJS74}. We will now show that an even more restricted variant of vertex cover is NP-hard, and then describe a reduction from it to EAP.
    
    We say that a graph $G=(V, E)$ is \emph{$2$-path-colored} if
    we can partition $E$ as $E=E_1\cup E_2$ so that
    \begin{enumerate}
        \item For each $h\in \{1,2\}$, each connected component of $(V,E_h)$ is an edge or a two-edge path; 
        \item Each $v\in V$ is incident to at least one edge in each of $E_1$ and $E_2$, and to at most three edges in total;
        \item $|E_1|=|E_2|$.
    \end{enumerate}
    In the problem {\sc Vertex Cover on $2$-path-colored graphs (VC-2PC)} we are given a $2$-path-colored graph $G$ together with the partition $E=E_1\cup E_2$ and a parameter $k$; the goal is to decide if $G$ admits a vertex cover of size~$k$.
    Then, we have the following proposition; of which the proof is deferred to the Appendix.
    \begin{proposition}
      {\sc VC-2PC} is {\em NP}-hard.
    \end{proposition}
    
    Then, to prove NP-hardness of EAP, we describe a reduction from {\sc VC-2PC}. 
    Given a $2$-path-colored graph $G=(V, E)$
    with $|V| = n$, $|E| = 2m$, partition $E = E_1 \cup E_2$, and a parameter $\kappa$,
    we construct an instance of EAP containing $2m$ proposals, 
    so that each proposal is approved by a narrow majority.
    
    For $h=1, 2$, we denote the edges of $E_h$ by $e_{h, 1}, \dots, e_{h, m}$.
    We number them so that, if $e_{h,j}$ and $e_{h,j'}$ share a vertex, then $|j-j'|=1$. We create one proposal $p_{h,j}$ for each edge $e_{h,j}\in E_1\cup E_2$. The proposals are ordered 
    as 
    $$
    p_{1,1}\prec \ldots  \prec p_{1,m} \prec 
    p_{2,1}\prec \ldots  \prec p_{2,m}.
    $$
    Let $P$ denote the set of all proposals.
    
    We introduce two sets of agents: the \emph{coding set} contains $n$ agents encoding the graph (one for each vertex of $V$) and the \emph{balancing set} contains $O(nm)$ agents that help set the majority correctly on each proposal. 
    
    Specifically, to build the coding set, 
    we create an agent $i$ for each vertex $v_i\in V$. The set $P^?(i)$
    consists of proposals corresponding to the edges incident to $v_i$.
    Note that $P^?(i)$ forms at most two intervals of the order $\prec$. Indeed,  this is immediate if $v_i$ has degree two. If $v_i$ has degree three, then it is the center of a $2$-edge path in some $E_h$, and the edges of this path are numbered as $e_{h, j}, e_{h_{j+1}}$ for some $j\in [m-1]$. The set $P^+(i)$ consists of all proposals 
    in between the leftmost proposal in $P^?(i)$
    and the rightmost proposal in $P^?(i)$ (inclusive).
    
    Let $Z=\max_{p\in P}|N^+(p)|$. 
    We build the balancing set in two steps.
    First, for each $p\in P$ we create
    $Z-|N^+(p)|$ agents $i$ with $P^+(i)= \{p\}$, $P^?(i)=\varnothing$.
    After this step, for each $p\in P$ we have $|N^+(p)|=Z$, and 
    exactly $2$ agents in $N^+(p)$ are uncertain about $p$. 
    If there are $X$ agents in total at that point, 
    then for each $p\in P$
    we have $|N^-(p)|=X-Z$; let $Y=X-Z$. At the second step, 
    if $Z< Y+3$, then we add $Y+3-Z$ agents $i$ with $P^{+!}(i)=P$, and
    if $Z>Y+3$, then we add $Z-3-Y$ agents $i$ with $P^{-!}(i)=P$.
    After this step, for each proposal $p\in P$, we have 
    $|N^+(p)| = |N^-(p)| +3$;
    two agents in $N^+(p)$ are uncertain about $p$.
    
    Thus, even though each proposal $p\in P$ is good, 
    it may receive $|N^-(p)|+2$ approvals 
    and $|N^+(p)|-2<|N^-(p)|+2$ disapprovals, 
    so $0$ is a possible outcome. However, if one of the agents who are uncertain about $p$ is educated, then at least $|N^+(p)|-1$ agents approve $p$ and at most $|N^-(p)|+1<|N^+(p)|-1$ agents disapprove it, so $1$ becomes the only possible outcome. That is, to make safe a proposal $p$ that corresponds to an edge $e$ we need to educate one of the two agents that correspond to the endpoints of $e$. Hence, we can make all proposals safe by educating $\kappa$ agents if and only if $G$ admits a vertex cover of size $\kappa$.
\end{proof}

Proving NP-hardness for DAP is similar to that for EAP (albeit with a slight modification), and hence we defer the proof of the following Theorem to the Appendix.

\begin{theorem}\label{thm:1D-DAP_is_hard}
  {\sc DAP} is {\em NP}-complete, even in the one-di\-men\-si\-onal 
  proposal space.
\end{theorem}

Finally, we prove NP-hardness for CLD with the following Theorem.

\begin{theorem}\label{thm:1D-CLD_is_hard}
  {\sc CLD} is {\em NP}-complete, even in the one-di\-men\-si\-onal 
  proposal space.
\end{theorem}
\begin{proof}
    We reduce from the NP-hard problem {\sc Vertex Cover on Cubic Graphs (VC-CG)}~\cite{GJS74}. Figure~\ref{fig:cldnpfig} (in the Appendix) illustrates an example of the reduction.
    
    An instance of this problem consists of an undirected cubic graph $G=(V, E)$ (i.e., the degree of each vertex $v\in V$ is exactly $3$) and an integer $\kappa$; it is a yes-instance if $G$ admits a vertex cover of size~$\kappa$ (i.e., a subset $V'\subseteq V$ that contains at least one endpoint of each edge), and a no-instance otherwise. 

    Consider an instance of VC-CG given by a graph $G = (V,E)$ and an integer $\kappa$. 
    
    Fix an arbitrary order $\lhd_E$ on $E$.
    For each edge $e \in E$ with endpoints $u$ and $v$, we introduce proposals 
    $s_e, r_{e,u}, p_{e}, r_{e,v}$ (jointly denoted as the \emph{edge gadget for $e$}). For each vertex $u$, we introduce proposals $s_u$, $p_u$, and $q_u$ (jointly denoted as the \emph{vertex gadget for $u$}). Additionally, we introduce two proposals, denoted by $s_\$, p_\$$ (the \emph{sink gadget}).

    The proposals are ordered so that 
    (1) all edge gadgets precede all vertex gadgets, and all vertex gadgets precede the sink gadget;
    (2) the edge gadgets appear in the order induced by $\lhd_E$;
    (3) the vertex gadgets appear in an arbitrary order;
    (4) within an edge gadget for $e=\{u, v\}$, the proposals are ordered
    as $s_e\prec r_{e,u}\prec p_{e}\prec r_{e,v}$;
    (5) within a vertex gadget for $u$, the proposals are ordered 
    as $s_u\prec p_u\prec q_u$;
    (6) within the sink gadget, the proposals are ordered 
    as $s_\$\prec p_\$$.
    
    For each vertex $u\in V$,
    we consider the edges $a$, $b$, $c$ incident on $u$ (assume $a\lhd_E b \lhd_E c$), and introduce eight agents, denoted $v^u_{a,1},v^u_{a,2}, v^u_{b,1},v^u_{b,2},v^u_{c,1},v^u_{c,2}, v^u_1, v^u_2$ (they are said to be \emph{related to $u$}) 
    with the following beliefs:
    \begin{itemize}
        \item the sets $P^+(v^u_{a,1})$ and $P^+(v^u_{a,2})$ consist of all proposals in the shortest interval containing $p_a, r_{a,u}$ and $r_{b,u}$, and $P^?(v^u_{a,1})= \{p_a, r_{a,u}\}$, whereas $P^?(v^u_{a,2}) = \{r_{b,u}\}$.
        \item the beliefs of $v^u_{b,1}$ and $v^u_{b,2}$ (resp., $v^u_{c,1}$ and $v^u_{c,2}$) are defined similarly, replacing $p_a$ by $p_b$ (resp., $p_c$), $r_{a,u}$ by $r_{b,u}$ (resp., $r_{c,u}$) and $r_{b,u}$ by $r_{c,u}$ (resp., $p_u$).
        \item the sets $P^+(v^u_1)$ and $P^+(v^u_2)$ consist of all proposals in the interval from $p_u$ to $p_\$$, and $P^?(v^u_1) =\{p_u, q_u\}$, $P^?(v^u_2) = \{p_\$\}$.
    \end{itemize}
    We then add polynomially many \textit{balancing agents} so that we have $2K+1$ agents in total 
    for some $K\in\mathbb N$ and:
    \begin{itemize}
        \item 
        $|N^+(p_e)|=K+2$, $|N^{+?}(p_e)|=2$ for each $e\in E$;
        \item 
        $|N^+(r_{e, u})|=K+2$, $|N^{+?}(r_{e, u})|=1$ or 
        $|N^+(r_{e, u})|=K+3$, $|N^{+?}(r_{e, u})|=2$
        for each $e\in E$, $u\in e$;
        \item 
        $|N^+(p_u)|=|N^+(q_u)| = K+3$, $|N^{+?}(p_u)|=2$, $|N^{+?}(q_u)|=1$;
        \item 
        $|N^+(s_e)|=|N^+(s_u)| = |N^+(s_\$)| = K+1$, 
        $N^?(s_e)=N^?(s_u) = N^?(s_\$) = \varnothing$,
        for each $e\in E$, $u\in V$; and 
        \item 
        $|N^+(p_\$)|= K + |V| + \kappa +1$, $|N^{+?}(p_\$)|=|V|$. 
    \end{itemize}
    To this end, we add agents that benefit from a single proposal only (and are certain about it), until the desired differences in support for each proposal are accomplished, and then add sufficiently many agents who benefit from all proposals, to get the target numbers.
    Let $N_V$ denote the set of all vertex-related agents and let $N_B$ be the set of balancing agents.
    
    We now consider the delegation graph. Note first that for each
    $u\in V$ agent $v^u_1$ may delegate to $v^u_2$, and for each 
    edge $e$ incident on $u$ agent $v^u_{e,1}$ may delegate to $v^u_{e,2}$. We claim that the delegation graph cannot contain
    any other arcs.
    
    To see this, observe first that 
    for each $i\in N_B$ we have $P^?(i)=\varnothing$, $|P^+(i)|=1$.
    Moreover, for each agent $j\in N_V$ we have $|P^{+!}(j)|\ge 2$.
    Hence, no agent in $N_B$ can delegate her vote to another agent,
    and no agent in $N_V$ can delegate her vote to an agent in $N_B$.
    
    Further, let $S=\{s_\xi: \xi\in V\cup E\cup\{\$\}\}$ be the set of {\em separators}, and note that $N^?(s) = \varnothing$ for all $s\in S$. Hence, for the arc $(i, j)$ to be in the delegation graph, $i$ and $j$ must agree on all separators.
    This can only happen if the positive intervals for $i$ and $j$ start in the same gadget as well as end in the same gadget, 
    which establishes our claim.
    
    It follows that each path in the delegation graph consists of at most one arc. Hence, any delegation scenario may only change the number of uncertain agents for each proposal, e.g., for each $p \in P$, no agent in $N^{+!}(p)$ can delegate (either directly or indirectly) to an agent in $N^{-!}(p)$ or vice versa.
     
     Hence, we can make all proposals safe if and only if the following conditions are satisfied: 
     \begin{itemize}
         \item For each $e\in E$, $|N^?(p_e)|$
         decreases by at least 1.
         \item For each $e\in E$ and $u\in e$, $|N^?(r_{e,u})|$
         does not increase.
         \item For each $u \in V$, $|N^?(p_u)|$ 
         does not increase.
         \item $|N^?(p_\$)|$ increases by at most $\kappa$.
     \end{itemize}
    
    We will now prove that there exists a delegation graph $\calD$ that guarantees correct decisions on all proposals if and only if $G$ has a vertex cover of size $\kappa$.
    
    For the `if' direction, let $X$ be a size-$\kappa$ vertex cover of $G$. 
    Then, for each $x \in X$, consider the following delegations: $v^x_{e,1} \rightarrow v^x_{e,2}$ for each $e$ incident on $x$,
    and $v^x_1 \rightarrow v^x_2$. As each edge $e$ is incident on some $x\in X$, these delegations reduce by $1$ the size of each set $N^?(p_e)$, $e\in E$, and increase by $\kappa$ the size of the set $N^?(p_\$)$. The only other proposals affected by these delegations are proposals of the form $r_{e, x}$ and $p_x, q_x$, where $x\in X$. But then it can be checked that the sizes of the sets $N^?(r_{e, x})$, $N^?(p_x)$ and $N^?(q_x)$ either decrease by one or do not change. Specifically, if $a\lhd_E b\lhd_E c$ are the three edges incident on $x$ then 
    (1) $v^x_{a,1} \rightarrow v^x_{a,2}$
    decreases $|N^?(r_{a, x})|$ by $1$ and
    increases $|N^?(r_{b, x})|$ by $1$;
    (2) $v^x_{b,1} \rightarrow v^x_{b,2}$
    decreases $|N^?(r_{b, x})|$ by $1$ and
    increases $|N^?(r_{c, x})|$ by $1$;
    (3) $v^x_{c,1} \rightarrow v^x_{c,2}$
    decreases $|N^?(r_{c, x})|$ by $1$ and
    increases $|N^?(p_x)|$ by $1$;
    (4) $v^x_{1} \rightarrow v^x_{2}$
    decreases $|N^?(p_x)|$ and $|N^?(q_x)|$ by $1$.
    Overall, the number of uncertain agents for each proposal falls within the desired bounds, and all proposals are safe.
    
    For the `only if' direction, 
    consider a delegation graph $\calD$ that guarantees correct decisions on all proposals. Fix a node $u\in V$, and let $a\lhd_E b\lhd_E c$ be the edges of $G$ incident on $u$. If $\calD$ contains the arc $v^u_{c, 1} \rightarrow v^u_{c, 2}$, then 
    it must also contain the arc
    $v^u_1\rightarrow v^u_2$ (or else $|N^?(p_u)|$ would increase by 1).
    Similarly, if $\calD$ contains the arc
    $v^u_{b, 1} \rightarrow v^u_{b, 2}$, it must also contain 
    $v^u_{c, 1} \rightarrow v^u_{c, 2}$ (or else $|N^?(r_{c, u})|$ would increase by 1) and hence
    $v^u_1\rightarrow v^u_2$.
    Finally, if $\calD$ contains the arc
    $v^u_{a, 1} \rightarrow v^u_{a, 2}$, it must also contain 
    $v^u_{b, 1} \rightarrow v^u_{b, 2}$ (or else $|N^?(r_{b, u})|$ would increase by 1) and hence
    $v^u_{c, 1} \rightarrow v^u_{c, 2}$ and, in turn, $v^u_1\rightarrow v^u_2$.
    That is, if $\calD$ contains a delegation between two agents related to $u$, it must contain the arc $v^u_1\rightarrow v^u_2$.
    
    Let $X$ be the set of vertices $x$ with $v^x_1 \rightarrow v^x_2$. Each such delegation increases $|N^?(p_\$)|$ by $1$, so $|X|\leq \kappa$.
    Now, consider an edge $e=\{u,w\}$. Since $|N^?(p_e)|$ must decrease by at least one, at least one of the delegations $v^u_{e,1} \rightarrow v^u_{e,2}$ or $v^w_{e,1} \rightarrow v^w_{e,2}$ must take place; assume without loss of generality that $\calD$ contains the arc $v^u_{e,1} \rightarrow v^u_{e,2}$. 
    As argued above, this means that $\calD$ also contains
    the arc $v^u_1 \rightarrow v^u_2$, i.e., $u\in X$. Hence, $X$ is a vertex cover for $G$ of size at most $\kappa$.
    \end{proof}

\section{Fixed-Parameter Tractability in $\boldsymbol{n}$ and $\boldsymbol{m}$}\label{sec:FPT}

In this section, we show that all three of our problems---EAP, DAP, and CLD---are fixed-parameter tractable (FPT) with respect to the number of voters ($n$), and independently, with respect to the number of proposals ($m$). These results hold even for the maximization version of these problems, where the goal is to maximize the number of proposals on which the correct outcome is obtained, and also do not require the assumption that the proposal space is one-dimensional. 

For the number of agents $n$, we can use brute-force algorithms, as there are $2^n$ possible subsets of agents to educate or delete, and at most $n^n$ different delegation graphs.

For the number of proposals $m$, our approach is based on integer linear programming; we show how to encode each of EAP, DAP, and CLD as an integer linear program (ILP) whose number of variables depends on $m$ (but not on $n$);
our claim then follows from Lenstra's classic result~\cite{Lenstra83}. 
To accomplish this, we classify the agents into `types', so that the number of types is exponential in $m$, but does not depend on $n$.

Our next result show that each of our problems are FPT in $m$ (the case for $n$ is trivial as detailed above). 


\begin{proposition}\label{prop:fpt-m}
{\sc EAP}, {\sc DAP}, and {\sc CLD} are {\sc FPT} for $n$, as well as {\sc FPT} for $m$.
\end{proposition}

\begin{proof}
    We only detail the proof for DAP---proofs for EAP and CLD can be found in the Appendix.
    
    First, note that we can remove all proposals $p$ with $|N^+(p)|=|N^-(p)|$
    from the description of our voting instance (as we are satisfied with either outcome on any such proposal).
    Thus, from now on we will assume that, for each $p\in P$, we have
    $|N^+(p)|>n/2$ or $|N^-(p)|>n/2$.
    
    Given an agent $i$, construct a string $\vect = (t_1, \dots, t_m)$ 
    over $\{+, -, ?\}$, where, for each $j\in [m]$, we set 
    $t_j=+$ if $p_j\in P^{+!}(i)$, 
    $t_j=-$ if $p_j\in P^{-!}(i)$, and
    $t_j=?$ if $p_j\in P^{?}(i)$. We will refer 
    to $\vect$ as the {\em type} of $i$.
    Let $T = \{+, -, ?\}^m$.
    By construction, there are at most $3^m$
    distinct agent types. 
    
    For each $\vect\in T$, 
    let $y_\vect$ denote the number of agents of type $\vect$ in $I$, and let $x_\vect$ denote the number of agents of type $\vect$ in the instance $I'$ obtained after some agents have been deleted. The constraint that we can remove
    at most $\kappa$ agents is encoded as
    \begin{equation}\label{eq:dap-count}
        0\le x_\vect\le y_\vect\text{\ \ for all \ \ }\vect\in T, \quad 
        \sum_{\vect\in T}x_\vect \ge n-\kappa. 
    \end{equation}
    
    For each proposal $p_j\in P$, let 
    $T^{+!}(j) = \{\vect\in T: t_j=+\}$, 
    $T^{-!}(j) = \{\vect\in T: t_j=-\}$, 
    $T^{?}(j) = \{\vect\in T: t_j=?\}$.
    Given a proposal $p_j$, let
    \begin{align}\label{eq:dap-ab}
    a_j &= 
    \sum_{\vect\in T^{+!}(j)}x_\vect - \sum_{\vect\in T^{-!}(j)}x_\vect -  \sum_{\vect\in T^{?}(j)}x_\vect; \\
    b_j &= 
    \sum_{\vect\in T^{-!}(j)}x_\vect - \sum_{\vect\in T^{+!}(j)}x_\vect -  \sum_{\vect\in T^{?}(j)}x_\vect.\nonumber
    \end{align}
    Note that $a_j, b_j\le n$.
    If $p_j$ is good, then we would like $a_j$ to be positive, and 
    if $p_j$ is bad, then we would like $b_j$ to be positive.
    Thus, for each $j\in [m]$ we introduce a variable $z_j$
    that takes values in $\{0, 1\}$, and add the constraint
    \begin{align*}
    n\cdot z_j \ge a_j, \qquad n\cdot(1-z_j)\ge 1-a_j\qquad\qquad(\text{$j$-good})
    \end{align*}
    if $p_j$ is good and the constraint
    \begin{align*}
    n\cdot z_j \ge b_j, \qquad n\cdot(1-z_j)\ge 1-b_j\qquad\qquad(\text{$j$-bad})
    \end{align*}
    if $p_j$ is bad.
    The reader can verify that condition ($j$-good)
    ensures that $z_j=1$ if and only if $a_j>0$, whereas
    condition ($j$-bad) ensures that $z_j=1$ if and only if $b_j>0$.
    
    To summarize, the set of variables of our ILP is
    $\{x_\vect: \vect\in T\}\cup\{a_j, b_j, z_j: j\in [m]\}$,
    and its objective function is
    $\max \sum_{j\in [m]} z_j$.
    The set of constraints consists of~\eqref{eq:dap-count}, \eqref{eq:dap-ab}, 
    and, for each $j\in [m]$, one of the constraints ($j$-good) or ($j$-bad), 
    depending on whether $p_j$ is good or bad, together with the constraint $0\le z_j\le 1$.
    Following Lenstra's result~\cite{Lenstra83},
    this ILP can be solved in time $\mathrm{poly}\brackets{\brackets{3^m}^{\calO\brackets{3^m}}, n}$.
\end{proof}

\section{The Radical One-Dimensional Domain}\label{sec:rad}

Here we consider a more restricted variant of the one-dimensional model in which, for each agent $i\in N$, the set $P^+(i)$ is a suffix of $\prec$, 
i.e., either $P^+(i)=\varnothing$ or $p_m\in P^+(i)$. Intuitively, 
this model reflects settings where the proposals can be naturally ordered from 
radical proposals to mild proposals, so that a typical agent disapproves the 
most radical proposal, but approves the mildest proposal, and switches
from disapproval to approval at a certain point in which the proposals become less extreme.
Consequently, in this model, referred to as the \emph{radical 1D domain}, $P^?(i)$ is an interval of $\prec$, for each $i\in N$. 
For convenience, in what follows we assume that the number of agents 
$n$ is odd and hence $\corr(p, I)$ is a singleton for each $p\in P$ (but our results hold also for even number of agents);
slightly abusing notation, we write $\corr(p, I)=z$ instead of $\corr(p, I)=\{z\}$.

First, we make the following two observations.

\begin{observation}\label{observation:1dplus}
  For the radical 1D domain, the list $(\corr(p_i, I))_{i\in [m]}$
  is of the form $(0, \dots, 0, 1, \dots, 1)$.
\end{observation}

\begin{proof}
Suppose that $\corr(p_j, I)=1$ for some $j<m$. Then $|N^+(p_j)|>n/2$. Consider a proposal $p_k$ with $k>j$. For each agent $i\in N$ with $p_j\in P^+(i)$ we have $p_k\in P^+(i)$ and hence $N^+(p_j)\subseteq N^+(p_k)$. It follows that
$|N^+(p_k)|>n/2$, i.e., $\corr(p_k, I)=1$.
\end{proof}

\begin{observation}\label{observation:leftright}
  For the radical 1D domain, if a proposal $p_j$ is safe and $\corr(p_j, I)=0$,
  then each proposal $p_k$ with $k\le j$ is safe. Similarly, if 
  $p_j$ is safe and $\corr(p_j, I)=1$, then each proposal $p_k$ with $k\ge j$ is safe.
\end{observation}

\begin{proof}
Fix a safe proposal $p_j$. Suppose $\corr(p_j, I)=0$;
the case $\corr(p_j, I)=1$ is symmetric. As $p_j$ is safe, we have 
$|N^{-!}(p_j)|>n/2$. Consider a proposal $p_k$ with $k<j$.
Note that $p_j\in P^{-!}(i)$ implies $p_k\in P^{-!}(i)$.
Hence $N^{-!}(p_j)\subseteq N^{-!}(p_k)$ and
$|N^{-!}(p_k)|\ge |N^{-!}(p_j)|>n/2$, so $p_k$ is safe.
\end{proof}

We will now show that for the radical 1D domain the problems 
EAP, DAP and CLD admit polynomial-time algorithms.

\begin{proposition}\label{prop:poly-MAX-1D+}
  In the radical 1D domain, {\em EAP}, {\em DAP}, and {\em CLD} are all solvable in polynomial time.
\end{proposition}

\begin{proof}
By Observation~\ref{observation:leftright}, unsafe proposals form an interval of the form $\{p_{j+1}, \dots, p_{k-1}\}$ for some $p_j, p_k\in P$ with $\corr(p_j, I)=0$, $\corr(p_k, I)=1$.
We consider all $O(m^2)$ possible choices for $j$ and $k$ such that $k+1-j\le \lambda$, and aim to modify the instance so 
as to guarantee correct outcomes on $p_j$ and $p_k$;
this ensures that we obtain
correct outcomes on all proposals, except possibly for the proposals 
in the set $\{p_{j+1}, \dots, p_{k-1}\}$. 

For EAP and DAP, we can decide if a suitable modification exists 
by running our FPT algorithm (Propositions~\ref{prop:fpt-m}) on the set of proposals $\{p_j, p_k\}$;
note that this algorithm runs in polynomial time when there are just two proposals.

Unfortunately, this strategy fails for CLD, 
because proposals other than $p_j$ and $p_k$ 
influence which edges can appear in a delegation graph. 
Hence, we use a different approach. 
We say that an agent $i$ is {\em $I$-correct
on proposal $p_\ell$} 
if $\corr(p_\ell, I) = 0$ and $p_\ell\in P^{-!}(i)$, 
or $\corr(p_\ell, I) = 1$ and $p_\ell\in P^{+!}(i)$.
Also, we say that a delegation graph $\calD$ is {\em $(j, k)$-good} if in the instance $I'(\calD)$ more than $n/2$ agents are $I$-correct on $p_j$
and more than $n/2$ agents are $I$-correct on $p_k$.

To check if there is a $(j, k)$-good delegation graph, 
we partition the set $N$ as 
$N = N_A\cup N_B\cup N_C\cup N_D$ so that
\begin{itemize}
    \item agents in $N_A$ are $I$-correct on both $p_j$ and $p_k$;
    \item agents in $N_B$ are $I$-correct on $p_j$, but not on on~$p_k$;
    \item agents in $N_C$ are $I$-correct on $p_k$, but not on~$p_j$; and
    \item agents in $N_D$ are $I$-correct on both $p_j$ and $p_k$.
\end{itemize}
Suppose there exists a $(j, k)$-good delegation graph $\calD$.
Then, in $I'(\calD)$,
more than $n/2$ agents are $I$-correct on $p_j$ and
more than $n/2$ agents are $I$-correct on $p_k$. 
As these two sets of agents must overlap, there is an agent in $I'(\calD)$
that is $I$-correct on both $p_j$ and $p_k$.
Since delegation amounts to agents
copying the preferences of their gurus, this means that there is some such agent in $I$, i.e., $N_A$ must be non-empty.

Thus, if $N_A=\varnothing$, 
there is no $(j, k)$-good delegation graph.
Assume, then, that $N_A$ is non-empty.
Note that all agents in $N_D$ must be uncertain about $p_j$ and $p_k$
(and hence about all proposals in between) and are therefore 
willing to delegate to agents in $N_A$, who are certain about both.
Further, if there is a $(j, k)$-good delegation graph, then there is 
one in which all agents who can delegate to an agent in $N_A$ do so
(and thereby copy the preferences of their guru).
Hence, as our next step, we repeatedly check whether $N_B\cup N_C\cup N_D$
contains an agent who is willing to delegate to an agent in $N_A$;
if so, then we move this agent to $N_A$. This process stops when no such agent
remains. Observe that at this point $N_D=\varnothing$.

Now, if $|N_B|+|N_A|>n/2$,
then there is a majority of agents who are $I$-correct on $p_j$, and 
if $|N_C|+|N_A|>n/2$, then there is a majority of agents who are $I$-correct on $p_k$; so, if both of these conditions are satisfied, then we are done.
Assume without loss of generality that $|N_C|\ge |N_B|$. 
As we have $|N_B|+|N_A|+|N_C|+|N_A|=n+|N_A|>n$, the condition
$|N_C|+|N_A|>n/2$ is satisfied in this case. 

Suppose, however, that $|N_B|+|N_A|\le n/2$.
Then, the only way to obtain the correct outcome
on both $p_j$ and $p_k$ is to find 
$t = \lceil\frac{n}{2}\rceil-|N_A|-|N_B|$ agents in $N_C$
who can delegate to agents in $N_B$; indeed, 
in this case, after the delegation we will have $\lceil\frac{n}{2}\rceil$
agents who are $I$-correct on $p_j$ and 
$|N_C|-t+|N_A| = n+|n_A|-\lceil\frac{n}{2}\rceil \ge \lceil\frac{n}{2}\rceil$
agents who are $I$-correct on $p_k$ (where the last inequality
holds because $|N_A|\ge 1$). 

To determine whether these $t$ delegating agents could be found, 
we construct a graph $G$ on $N_B\cup N_C$
with an edge from agent $i$ to agent $i'$
if and only if $i$ is willing to delegate to $i'$, 
and mark each agent in $N_C$ that has a path to an agent in $N_B$;
the answer is positive if and only if at least $t$ agents are marked.
This completes the proof.
\end{proof}


\section{Outlook}

We have considered three algorithmic approaches to coping with confused communities that wish to reach good joint decisions, and showed that, even though our problems are generally NP-complete, there are special cases for which efficient algorithms exist.

Avenues for future research include the following:
(1) a study of other proposal spaces, including, in particular, continuous proposal spaces;
(2) extending our analysis to settings where proposals may have different importance and/or issues are non-binary;
(3) an exploration of different approaches to coping with voter uncertainty, including models in which the effect of education efforts is probabilistic and may propagate through an underlying social network;
(4) the development of computer-based simulations for the practical evaluation of the effectiveness of our algorithmic approaches.

\bibliographystyle{named}
\bibliography{abb,bib}

\clearpage

\appendix

\begin{center}
    {\LARGE Appendix}
\end{center}

\section{Proofs Missing from Section~\ref{sec:NP}}


We show that EAP, DAP, and CLD are NP-hard even in the one-dimensional case.

\subsection{Proof of Theorem~\ref{thm:1D-EAP_is_hard}}

We show that EAP is NP-hard even in the one-dimensional case.

To this end, we provide a reduction from {\sc Vertex Cover on Cubic Graphs (VC-CG)}. An instance of this problem comprises of  an undirected cubic graph $G=(V, E)$ (i.e., the degree of each vertex $v\in V$ is exactly $3$) and an integer $\kappa$; it is a yes-instance if $G$ admits a vertex cover of size~$\kappa$ (i.e., a subset $V'\subseteq V$ that contains at least one endpoint of each edge), and a no-instance otherwise. This problem is known to be NP-complete~\cite{GJS74}. We will now show that an even more restricted variant of vertex cover is NP-hard, and then describe a reduction from it to EAP.

We say that a graph $G=(V, E)$ is \emph{$2$-path-colored} if
we can partition $E$ as $E=E_1\cup E_2$ so that
\begin{enumerate}
    \item For each $h\in \{1,2\}$, each connected component of $(V,E_h)$ is an edge or a two-edge path; 
    \item Each $v\in V$ is incident to at least one edge in each of $E_1$ and $E_2$, and to at most three edges in total;
    \item $|E_1|=|E_2|$.
\end{enumerate}
In the problem {\sc Vertex Cover on $2$-path-colored graphs (VC-2PC)} we are given a $2$-path-colored graph $G$ together with the partition $E=E_1\cup E_2$ and a parameter $k$; the goal is to decide if $G$ admits a vertex cover of size~$k$.

\begin{proposition}
  {\sc VC-2PC} is {\em NP}-hard.
\end{proposition}

\begin{proof}
We reduce from {\sc VC-CG}. Given a cubic graph $G=(V,E)$, we first color the edges of $G$ with four colors A, B, C, D so that no pair of edges of the same color have a common vertex; this is possible by Vizing's theorem~\cite{diestel}.

We now build a graph $G^*$ together with a red-blue edge coloring for it: we then place
all red edges into $E_1$ and all blue edges into $E_2$.
We start by creating two
disjoint copies of $G$, which we denote by $G'$ and $G''$.
We then replace each edge in these graphs with a $3$-edge path as follows.
If an edge $u$--$v$ in $G$ is colored A or B, then we replace its copy in $G'$ (respectively, $G''$) with a red--blue--red
path (respectively, with a blue--red--blue path). 
If an edge $u$--$v$ in $G$ is colored C or D, then we replace its copy in $G'$ (respectively, $G''$) with a blue--red--blue
path (respectively, with a red--blue--red path).

The graph $G^*=(V^*, E^*)$ has $2|V|+4|E|$ vertices and $6|E|$ edges. Each edge of $G$
corresponds to three red edges and three blue edges of $G^*$, so $|E_1|=|E_2|=3|E|$.
Moreover, since $G$ is cubic, the degree of each vertex in $G'$ and $G''$ is at most $3$. Also, since each vertex of $G$ is incident to at least one edge colored by A or B and to at least one edge colored by C or D, in $G^*$
each vertex is incident to at least one blue edge and at least one red edge.
By construction, $(V^*, E_1)$ is a collection of edges and $2$-edge paths, 
and so is $(V^*, E_2)$.

We will now argue that $G$ has a vertex cover of size $\kappa$ if and only if $G^*$ has a vertex cover of size $2(\kappa+|E|)$. It suffices to prove that $G$ has a
vertex cover of size $\kappa$ if and only if $G'=(V', E')$ has a vertex cover of size $\kappa+|E|$. 

Indeed, if $K\subseteq V$ is a vertex cover for $G$, then we can transform $K$ into a vertex cover
for $G'$ by adding one vertex per edge of $G$. Specifically, consider an edge $u$--$v$
of $G$, and suppose that it was replaced with a path $\rho_{uv}=u$--$x$--$y$--$v$. Since $K$ is a vertex cover for $G$, we have $\{u, v\}\cap K\neq\varnothing$. If $u\in K$, then, by adding $y$ to the cover, we ensure that all edges of $\rho_{uv}$ are covered. Similarly, if $v\in K$, then, by adding $x$ to the cover, we ensure that all edges of $\rho_{uv}$ are covered. Applying this transformation to each edge of $G$ results in a vertex cover for $G'$; the size of this cover is $|K|+|E|$. 

Conversely, suppose that $K'$ is a vertex cover for $G'$. Consider an edge $u$--$v$
of $G$ and the corresponding path $\rho_{uv}=u$--$x$--$y$--$v$ in $G'$. Since the edge
$x$--$y$ must be covered, we have $\{x, y\}\cap K'\neq\varnothing$. Furthermore, if
$\{u, v\}\cap K' = \varnothing$, then it has to be the case that both $x$ and $y$ are in 
$K'$ (as both edges $u$--$x$ and $y$--$v$ need to be covered). In this case, we can modify
$K'$ by removing $x$ and adding $u$; note that after this step $K'$ remains a vertex cover for $G'$. After applying this transformation to all edges, we ensure that $K'\cap V$
contains at least one endpoint of each edge of $G$ and hence forms a vertex cover of 
$G$. On the other hand, we have $|K'\setminus V|\ge |E|$, as even after the transformation $K'$ contains at least
one internal node of each $u$--$x$--$y$--$v$ path. Hence, $K=K'\cap V$ is a vertex 
cover for $G$ and $|K|\le |K'|-|E|$.
This completes the proof.
\end{proof}

\textbf{We are ready to show that EAP is NP-complete, even in the one-dimensional proposal space.}

\begin{proof}\emph{(of Theorem~\ref{thm:1D-EAP_is_hard})}
To prove NP-hardness, we describe a reduction from {\sc VC-2PC}. 
Given a $2$-path-colored graph $G=(V, E)$
with $|V| = n$, $|E| = 2m$, partition $E = E_1 \cup E_2$, and a parameter $\kappa$,
we construct an instance of EAP containing $2m$ proposals, 
so that each proposal is approved by a narrow majority.

For $h=1, 2$, we denote the edges of $E_h$ by $e_{h, 1}, \dots, e_{h, m}$.
We number them so that, if $e_{h,j}$ and $e_{h,j'}$ share a vertex, then $|j-j'|=1$. We create one proposal $p_{h,j}$ for each edge $e_{h,j}\in E_1\cup E_2$. The proposals are ordered 
as 
$$
p_{1,1}\prec \ldots  \prec p_{1,m} \prec 
p_{2,1}\prec \ldots  \prec p_{2,m}.
$$
Let $P$ denote the set of all proposals.

We introduce two sets of agents: the \emph{coding set} contains $n$ agents encoding the graph (one for each vertex of $V$) and the \emph{balancing set} contains $O(nm)$ agents that help set the majority correctly on each proposal. 

Specifically, to build the coding set, 
we create an agent $i$ for each vertex $v_i\in V$. The set $P^?(i)$
consists of proposals corresponding to the edges incident to $v_i$.
Note that $P^?(i)$ forms at most two intervals of the order $\prec$. Indeed,  this is immediate if $v_i$ has degree two. If $v_i$ has degree three, then it is the center of a $2$-edge path in some $E_h$, and the edges of this path are numbered as $e_{h, j}, e_{h_{j+1}}$ for some $j\in [m-1]$. The set $P^+(i)$ consists of all proposals 
in between the leftmost proposal in $P^?(i)$
and the rightmost proposal in $P^?(i)$ (inclusive).

Let $Z=\max_{p\in P}|N^+(p)|$. 
We build the balancing set in two steps.
First, for each $p\in P$ we create
$Z-|N^+(p)|$ agents $i$ with $P^+(i)= \{p\}$, $P^?(i)=\varnothing$.
After this step, for each $p\in P$ we have $|N^+(p)|=Z$, and 
exactly $2$ agents in $N^+(p)$ are uncertain about $p$. 
If there are $X$ agents in total at that point, 
then for each $p\in P$
we have $|N^-(p)|=X-Z$; let $Y=X-Z$. At the second step, 
if $Z< Y+3$, then we add $Y+3-Z$ agents $i$ with $P^{+!}(i)=P$, and
if $Z>Y+3$, then we add $Z-3-Y$ agents $i$ with $P^{-!}(i)=P$.
After this step, for each proposal $p\in P$, we have 
$|N^+(p)| = |N^-(p)| +3$;
two agents in $N^+(p)$ are uncertain about $p$.

Thus, even though each proposal $p\in P$ is good, 
it may receive $|N^-(p)|+2$ approvals 
and $|N^+(p)|-2<|N^-(p)|+2$ disapprovals, 
so $0$ is a possible outcome. However, if one of the agents who are uncertain about $p$ is educated, then at least $|N^+(p)|-1$ agents approve $p$ and at most $|N^-(p)|+1<|N^+(p)|-1$ agents disapprove it, so $1$ becomes the only possible outcome. That is, to make safe a proposal $p$ that corresponds to an edge $e$ we need to educate one of the two agents that correspond to the endpoints of $e$. Hence, we can make all proposals safe by educating $\kappa$ agents if and only if $G$ admits a vertex cover of size $\kappa$.
\end{proof}

\subsection{Proof of Theorem~\ref{thm:1D-DAP_is_hard}}


Our hardness proof for DAP (i.e., our proof for Theorem~\ref{thm:1D-DAP_is_hard}) is a modification of the reduction presented in Theorem~\ref{thm:1D-EAP_is_hard}.

\newcommand{\M}{\textup{M}}
\begin{proof}\emph{(of Theorem~\ref{thm:1D-DAP_is_hard})}
Compared to EAP, the additional difficulty in DAP results from the fact that deleting an agent affects not just the proposals on which she is uncertain, but also all other proposals, including those on which she has the majority opinion (so deleting an agent may have also negative effects, in contrast to educating an agent, which does not). To handle this, we introduce agents in groups of four agents each, so that in each group two agents benefit from the proposal and two do not. Intuitively, such groups must either be deleted entirely or retained fully, having a null net effect on all proposals, except those where one of the agents is uncertain. To enforce the null net effect, we create two copies of each proposal:
  for one copy the correct decision is $1$ while for the other it is $0$.
Then, if an agent who benefits from both proposals is deleted, then this must be balanced out by deleting an agent who does not benefit from either proposal.
Details follow.

Consider an instance of {\sc VC-2PC} given by a graph $G=(V,E)$ with $|V|=n$, $|E|=2m$, partition $E=E_1\cup E_2$, and an integer $\kappa$. 
As in the proof of Theorem~\ref{thm:1D-EAP_is_hard}, we number the edges of $E_1$ and $E_2$ 
as $e_{1, 1}, \dots, e_{1, m}$ and $e_{2, 1}, \dots, e_{2, m}$ 
so that the edges of each two-edge path in $E_1$ or $E_2$ have consecutive indices. 
Our instance of DAP has the following $2m+2$ proposals: 
\begin{itemize}
    \item for each $h\in\{1,2\}$ and each edge $e_{h,j}\in E_h$, we introduce two \emph{edge} proposals, denoted by $a_{h,j}$ and $b_{h,j}$;
    \item also, we introduce two \emph{middle} proposals, denoted by $a_\M,b_\M$.
\end{itemize}
We order the proposals as follows:
\begin{align*}
a_{1, 1}&\prec b_{1, 1} \prec a_{1, 2} \prec b_{1, 2} \prec\dots
\prec a_{1, m}\prec a_{2, m}\\
        &\prec a_\M\prec b_\M\prec a_{2, 1}\prec b_{2, 1}\prec\dots\prec 
a_{2, m}\prec b_{2, m}.
\end{align*}

The agents are partitioned into six sets, which we denote by $A,B,C,D,E,F$:
\begin{itemize}
    \item $A$ contains $\kappa$ agents, with $P^{+?}(i)=P$ for each $i\in A$;
    \item $B$ contains $2m+n-\kappa$ agents, with $P^{+!}(i)=P$ for each $i\in B$;
    \item For each vertex in $V$, we create one agent in $C$, one agent in $D$, and one agent in $E$. Fix a vertex $v_i\in V$, and let $e_{1, j}$ and $e_{2, j'}$ be the first edge of $E_1$ and the last edge of $E_2$ incident to $v_i$, respectively. 
    If $v_i$ has degree $3$, let $e_{h,j''}$ be 
    the remaining edge incident to $v_i$. We populate $C$, $D$, $E$ as follows:
    \begin{itemize}
        \item in $C$, we place an agent $c_i$ such that $P^+(c_i)$ is 
        the interval $[a_{1,j}, b_{2,j'}]$ and  
        $P^?(c_i)$ contains $a_{1,j}, b_{1,j}, a_{2,j'}, b_{2,j'}$ and, 
        if $v_i$ has degree 3, also $a_{h,j''}$ and $b_{h,j''}$;
        \item in $D$, we place an agent $d_i$ such that $P^+(d_i)$ is 
        the interval from $a_{1,1}$ (included) to $a_{1,j}$ (excluded), 
        and $P^?(c_i)=\varnothing$;
        \item in $E$, we place an agent $e_i$ such that
        $P^+(e_i)$ is
        the interval from $b_{2, j'}$ (excluded) to $b_{2,m}$ (included),
        and $P^?(i)=\varnothing$.
    \end{itemize}
    \item $F$ consists of the following $2m+1$ agents, with $P^{!}(f)=P$ for each $f\in F$:
    \begin{itemize}
        \item $f_\M$, with $P^+(f_M)=\{a_\M\}$; 
        \item for each edge $e_{h,j}$, two agents 
              $f^1_{h,j}$ and $f^2_{h,j}$, with $P^+(f^1_{h, j})=P^+(f^1_{h, j})=\{a_{h, j}\}$.
    \end{itemize}
\end{itemize}
Let $H=2m+2n$. Altogether, there are $\kappa+(2m+n-\kappa)+3n+2m+1=4m+4n+1= 2H+1$ agents, so $H+1$ agents constitute a majority.

Observe that each proposal is beneficial for each of the $2m+n$ agents in $A\cup B$ and for exactly one agent in $C_i\cup D_i\cup E_i$, for each $v_i\in V$. In addition, $a_\M$ is beneficial for one agent in $F$ (so for $H+1$ agent in total), 
and each $a_{h, j}$, where $h=1, 2$, $j\in [m]$, is beneficial for two agents in $F$ (so for $H+2$ agents in total). Hence, $a_\M$ and all proposals in $\{a_{h, j}: h=1, 2, j\in [m]\}$ are good and all proposals in $\{b_{h, j}: h=1, 2, j\in [m]\}$ are bad. 

Moreover, for each agent $i$ we have $P^?(i)\subseteq P^+(i)$, so all bad proposals are safe. Furthermore, $N^{+?}(a_\M) = A$, 
so $|N^{+?}(a_\M)|=k$. Also, for each $e_{h, j}\in E$ the set $N^{+?}(a_{h, j})$
consists of $\kappa$ agents in $A$ and two agents in $C$ (namely, those that correspond to the endpoints of $e_{h, j}$), so $|N^{+?}(a_{h, j})|=\kappa+2$.
 
We now claim that that we can make all proposals safe 
by deleting at most $4\kappa$ agents if and only if $G$ has a vertex cover of size $\kappa$.

For the `if' direction, let $K$ be a size-$\kappa$ vertex cover of $G$. 
Delete the $4\kappa$ agents in $A \cup \bigcup_{i: v_i \in K}\{c_i,d_i,e_i\}$. 
For each proposal $p$, 
this removes exactly $2\kappa$ agents from $N^+(p)$ and exactly $2\kappa$ agents from $N^-(p)$. Thus, all bad proposals remain safe. Moreover, all agents who are uncertain about $a_\M$ are removed, 
so this proposal becomes safe, too. Now, consider an edge $e_{h,j}$. Since one endpoint of this edge is in $K$, one additional agent who is uncertain about $a_{h,j}$ is removed, so there remains at most one agent in 
$N^{+?}(a_{h, j})$, as well as $H-2\kappa+1$ agents in
$N^{+!}(a_{h, j})$ and $H-2\kappa-1$ agents in 
$N^{-!}(a_{h, j})$. Hence, $a_{h, j}$ becomes safe as well.

For the `only if' direction, suppose that we can make all proposals safe 
by deleting a set $X$ of at most $4\kappa$ agents. Let $N'$ be the set of remaining agents.
Let $K=\{v_i: c_i\in X\}$. We will show that $|K|\le \kappa$ and $K$ is a vertex cover for $G$. Given a pair of proposals $(a_x,b_x)$ (where $x=\M$ or $x=(h, j)$ for $h=1, 2$, $j\in [m]$), 
we define values 
$\alpha_x, \beta_x, \gamma_x, \delta_x$ as follows: 
\begin{align*}
\alpha_x  &= |\{i\in N': a_x, b_x\in P^{+!}(i)\}|;\\
\beta_x   &= |\{i\in N': a_x, b_x\in P^{+?}(i)\}|;\\
\gamma_x  &= |\{i\in N': a_x, b_x\in P^{-!}(i)\}|;\\
\delta_x  &= |\{i\in N': a_x\in P^{+!}(i), b_x\in P^{-!}(i)\}|.
\end{align*} 
Note that $\alpha_x+\beta_x+\gamma_x+\delta_x = 2H+1-|X|$.
After the agents in $X$ are removed, 
$a_x$ is good and safe, whereas $b_x$ is bad and safe, 
so we have the following inequalities:
\begin{align*}
    \alpha_x+\delta_x &\geq \gamma_x+\beta_x +1 \\
    \gamma_x+\delta_x &\geq \alpha_x+\beta_x +1.
\end{align*}
Combining both inequalities yields 
\begin{equation} \label{eq:smallbeta}
    \delta_x \geq \beta_x +1.
\end{equation}
Consider the middle proposals $a_\M$ and $b_\M$. 
Clearly, $\delta_\M\leq 1$ (before deletions, 
there is only one agent, namely, $f_\M$, 
that benefits from $a_\M$, but not from $b_\M$), so by \eqref{eq:smallbeta} $\beta_\M=0$ and $\delta_\M=1$. Thus, all agents in $A$ are deleted, and the agent $f_\M$ is not deleted. 

Let 
\begin{align*}
    x^+ &= |X\cap(A\cup B\cup C)|,\\
    x^- &= |X\cap(D\cup E\cup F\setminus\{f_\M\})|.
\end{align*}
Note that 
$a_\M, b_\M\in P^+(i)$ for each 
$i\in A\cup B\cup C$
and 
$a_\M, b_\M\in P^-(j)$ for each 
$j\in D\cup E\cup F\setminus\{f_\M\}$.
Thus, by removing $X$, we remove $x^+$ agents 
 from each of $N^+(a_\M)$ and $N^+(b_\M)$
 and $x^-$ agents from each of $N^-(a_\M)$ and $N^-(b_\M)$. 
 
 Since $|N^+(a_\M)|=H+1$, and $a_\M$ is safe after deletion, 
 we have $H+1-x^+ > H-x^-$.
 Since $|N^+(b_\M)|=H$, and $b_\M$ is safe after deletion, 
 we have $H+1-x^- > H-x^+$.
 Thus, we have $x^+ -1 < x^- < x^+ +1$
 and hence $x^+=x^-$. Since $|X|\le 4\kappa$, 
 we have $x^+\leq 2\kappa$. Furthermore, 
 we have argued that $\kappa$ agents are deleted from $A$, 
 so at most $\kappa$ agents are deleted from $C$, 
 and hence $|K|\leq \kappa$.

Consider the proposals $a_{h,j},b_{h,j}$ for some edge $e_{h,j}$. By~\eqref{eq:smallbeta}, since $\delta_{h,j}\leq 2$, we have $\beta_{h,j}\leq 1$, i.e., at most one undeleted agent is uncertain about these proposals. If $e_{h,j}$ has endpoints $v_{i}, v_{i'}$, then this means that at least one of the agents $c_i,c_{i'}$ is in $X$, i.e., one of $v_{i},v_{i'}$ is in $K$. We conclude that $K$ is a size-$\kappa$ vertex cover for $G$.
\end{proof}



\section{Proofs Missing from Section~\ref{sec:FPT}}


%
%

\subsection{Proof of Proposition~\ref{prop:fpt-m}}

\begin{proof}\emph{(of Proposition~\ref{prop:fpt-m})}

\paragraph{EAP}
Again, we remove all proposals $p$ with $|N^+(p)|=|N^-(p)|$
and assume that, for each $p\in P$, we have
$|N^+(p)|>n/2$ or $|N^-(p)|>n/2$.

Given an agent $i$, construct a string $\vect = (t_1, \dots, t_m)$ 
over $\{+!, -!, +?, -?\}$, where for each $j\in [m]$ we set 
$t_j=+!$ if $p_j\in P^{+!}(i)$, 
$t_j=-!$ if $p_j\in P^{-!}(i)$, 
$t_j=+?$ if $p_j\in P^{+?}(i)$, and
$t_j=-?$ if $p_j\in P^{-?}(i)$.
We will refer 
to $\vect$ as the {\em type} of $i$.
Let $T = \{+!, -!, +?, -?\}^m$.
By construction, there are at most $4^m$
distinct agent types. 

We say that a type $\vect$ is {\em unconfused}
if $t_j\in \{+!, -!\}$ for each $j\in [m]$.
Let $T'$ be the set of all unconfused types.
We now define a function $f$ that describes how 
an agent's type changes after she has been educated.
Given a type $\vect\in T$, we write $\vect' = f(\vect)$
if $\vect'\in T'$ and the following conditions hold for each $j\in [m]$:
(1) if $t_j\in\{+!, -!\}$ then $t'_j=t_j$;
(2) if $t_j=+?$ then $t'_j=+!$;
(3) if $t_j=-?$ then $t'_j=-!$.
By construction, if we educate an agent of type $\vect$, 
her type becomes $f(\vect)$.

For each $\vect\in T$, 
let $y_\vect$ denote the number of agents of type $\vect$ in $I$, and let $x_\vect$ denote the number of agents of type $\vect$ in the instance $I'$ obtained after some agents have been educated. Note that, as we educate agents, the number of agents of each type $\vect\in T\setminus T'$ does not increase. Moreover, the increase in the number of agents of some type $\vect'\in T'$ is equal to the decrease in the total number of agents that belong to one of the types $\vect\in T\setminus T'$ with $f(\vect)=\vect'$; in other words, for each $\vect'\in T'$, the total number of agents that belong to
one of the types $\vect$ with $f(\vect)=\vect'$ does not change.
Moreover, we can educate at most $\kappa$ agents.
These constraints are captured as follows:
\begin{align}\label{eq:eap-count}
    0 &\le x_\vect\le y_\vect\text{ for all }\vect\in T\setminus T'\\
    \sum_{\vect\in T: f(\vect)=\vect'}(y_\vect-x_\vect) &= 0\quad\text{for each $\vect'\in T'$.}\nonumber\\
    \sum_{\vect\in T'}(x_\vect-y_\vect) &\le \kappa\nonumber
\end{align}
For each proposal $p_j\in P$, let 
$T^{+!}(j) = \{\vect\in T: t_j=`+?'\}$, 
$T^{-!}(j) = \{\vect\in T: t_j=`-?'\}$, 
$T^{?}(j) = \{\vect\in T: t_j\in\{`+?', `-?'\}$.
We will now proceed as in the proof for DAP:
  for each proposal, we define conditions that ensure
  that the outcome of this proposal is $1$ or $0$, 
  and then define binary variables that
  indicate which proposals are safe.
The conditions are essentially the same as in the proof for DAP;
we reproduce them here for readability.

Given a proposal $p_j$, let
\begin{align}\label{eq:eap-ab}
a_j &= 
\sum_{\vect\in T^{+!}(j)}x_\vect - \sum_{\vect\in T^{-!}(j)}x_\vect -  \sum_{\vect\in T^{?}(j)}x_\vect; \\
b_j &= 
\sum_{\vect\in T^{-!}(j)}x_\vect - \sum_{\vect\in T^{+!}(j)}x_\vect -  \sum_{\vect\in T^{?}(j)}x_\vect.\nonumber
\end{align}
For each $j\in [m]$ we introduce a variable $z_j$
that takes values in $\{0, 1\}$, and add constraints
\begin{align*}
n\cdot z_j \ge a_j, \qquad n\cdot(1-z_j)\ge 1-a_j\qquad\qquad(\text{$j$-good})
\end{align*}
if $p_j$ is good and
\begin{align*}
n\cdot z_j \ge b_j, \qquad n\cdot(1-z_j)\ge 1-b_j\qquad\qquad(\text{$j$-bad})
\end{align*}
if $p_j$ is bad.
Condition ($j$-good)
ensures that $z_j=1$ if and only if $a_j>0$, whereas
condition ($j$-bad) ensures that $z_j=1$ if and only if $b_j>0$.
Thus, our goal is to maximize $\sum_{j\in m}z_j$.

To summarize, the set of variables of our ILP for EAP is
$\{x_\vect: \vect\in T\}\cup\{a_j, b_j, z_j: j\in [m]\}$, 
its objective function is
$$
\max \sum_{j\in [m]} z_j, 
$$
and
the set of constraints consists of~\eqref{eq:eap-count}, \eqref{eq:eap-ab}, 
and, for each $j\in [m]$, one of the constraints ($j$-good) or ($j$-bad), 
depending on whether $p_j$ is good or bad and the constraint $0\le z_j\le 1$.

\paragraph{CLD}
First, construct a \textit{delegation graph} $G = (V, E)$, where each of the $3^m$ vertices in $V$ corresponds to a type of voter, and the edges are the pairwise delegations possible as defined in CLD (i.e., $(u,v) \in E \Leftrightarrow |v^?| < |u^?|$). For each $i \in [3^m]$, let $n_i$ be the number of agents of type $i$. Enumerate all possible paths into the set $P$ (this can be done in $\mathcal{O}((3^m-2)!)$).
    
    Let $P_i \subseteq P$ denote the tree with root at node $i$, constructed from all paths with source $i$. Note that this tree may have repeated edges.
    Let a path $p_{ij}^\ell \in P_i$ denote the delegation from a voter of type $i$ to another voter of type $j$ (eventually, the destination) via the $\ell^\text{th}$ branch in the tree. A path is a set of edges in $E$.
    
    For each $p_{ij}^\ell \in P$, create $m$ constants: $p_{ij1}^\ell, p_{ij2}^\ell, \dots, p_{ijm}^\ell$, where each variable $p_{ijk}^\ell \in \{-1,0,1\}$ for $k \in [m]$. This represents the number of votes one particular delegation of the form $p_{ij}^\ell$ will affect proposal $k$ in terms of the overall votes it receives. 
    
Next, for each proposal $k \in [m]$, create a variable $y_k \in \mathbb{Z}_{\geq 0}$ indicating how many initial voters' decision on proposal $k$ is correct. This can be done in $\mathcal{O}(nm)$.
    
    Let $x_{ij}^\ell$ be the number of delegations we make along path $p_{ij}^\ell$. The first constraint is that $\sum_{j,\ell} x_{ij}^\ell \leq n_i$, ensuring that the number of delegations never exceeds the number of agents of that type. Also, to ensure that there are sufficient number of voters of each type to split votes, it must hold that the number of non-zero weighted edges emerging from a node is less than the number of voters at that node. To model this, we create an indicator variable:
    \begin{equation}
        \sum_{e \in P_i} \mathbb{I}_{e>0} \leq n_i
    \end{equation}
    
    Then, for each proposal $k$, let
    \begin{equation}
        c_k = y_k + \sum_{i,j,\ell:p_{ij}^\ell \in P} (x_{ij}^\ell\times p_{ijk}^\ell)
    \end{equation}
    Intuitively, for each proposal $k$, the first term $y_k$ is the initial number of votes it has on the correct side, and the second term is the change to the number of votes for this particular project $k$ after our selection of delegations. So $c_k$ equals the number of (correct) votes that project $k$ has after delegation.
    
    Let the safety margin (towards the correct decision) of project $k$ be:
        \begin{equation}
            a_k = c_k - (n - c_k)
        \end{equation}
    Observe that $a_k \geq 0$ means that the project is safe, whereas $a_k < 0$ means that the project is not safe.
    
    Define a set of binary variables $\xi_1, \xi_2,\dots,\xi_m \in \{0,1\}$ such that $\xi_j = 1$ if and only if proposal $j$ is safe. Set the optimization objective for the ILP to be:
      $\max \sum_{i} \xi_i$.
    
    Define additional auxiliary variables $w_k \in \set{0,1}$.
    The constraint is as follows (there should be another constraint that makes sure number of delegations $x_{ij}$ should not exceed the number of agents of type $i$; again, can be pre-computed in pre-processing step): 
    \begin{equation}
            \begin{split}
                       \; \forall k : a_k \geq -\brackets{n+1}w_k + 1&, \; 0\leq \xi_k \leq 1-w_k
                    \end{split}
    \end{equation} 
    
    If $a_k \leq 0$, then this indicates that the proposal does not have sufficient (correct) votes to pass the threshold, and $w_k$ shall be $1$. This in turns leads to $\xi_k = 0$. If $a_k > 0$, then this indicates that the proposal has sufficient (correct) votes to pass the threshold, and hence $w_k = 0$, which makes it possible to set $\xi_k = 1$.
    
    In summary, the ILP has the objective function
    \begin{equation*}
        \max \sum_{i=1}^m \xi_i
    \end{equation*}
    subject to the following constraints:
    \begin{equation}
        \begin{split}
            (i)\; \forall i \in V : \sum_{j,\ell}x_{ij}^\ell \leq n \\
            (ii) \; \forall e \in E : e = \sum_{e \in p_{ij}^\ell} x_{ij}^\ell
            \\
            (iii)\; \forall i \in V : \sum_{e \in P_i} \mathbb{I}_{e>0} \leq n_i
            %
            \\
           (iv)\; \forall k : a_k \geq -\brackets{n+1}w_k + 1&
            \\
           (v)\; \forall k : 0\leq \xi_k \leq 1-w_k.
        \end{split}
    \end{equation}
This completes the proof.
\end{proof}

\begin{figure*}
    \centering
    \begin{tikzpicture}[every node/.style={draw, circle}]
    \node[fill=yellow] (u) at (0,0) {u};
    \node[fill=yellow] (y)at (-1,.5) {y};
    \node (z) at (-1,-.5) {z};
    \node (x) at (1,0) {x};    
    \draw (u)--(y)--(z)--(u)--(x);
    
    \end{tikzpicture}
    \vspace{0.3in}
    
\includegraphics[]{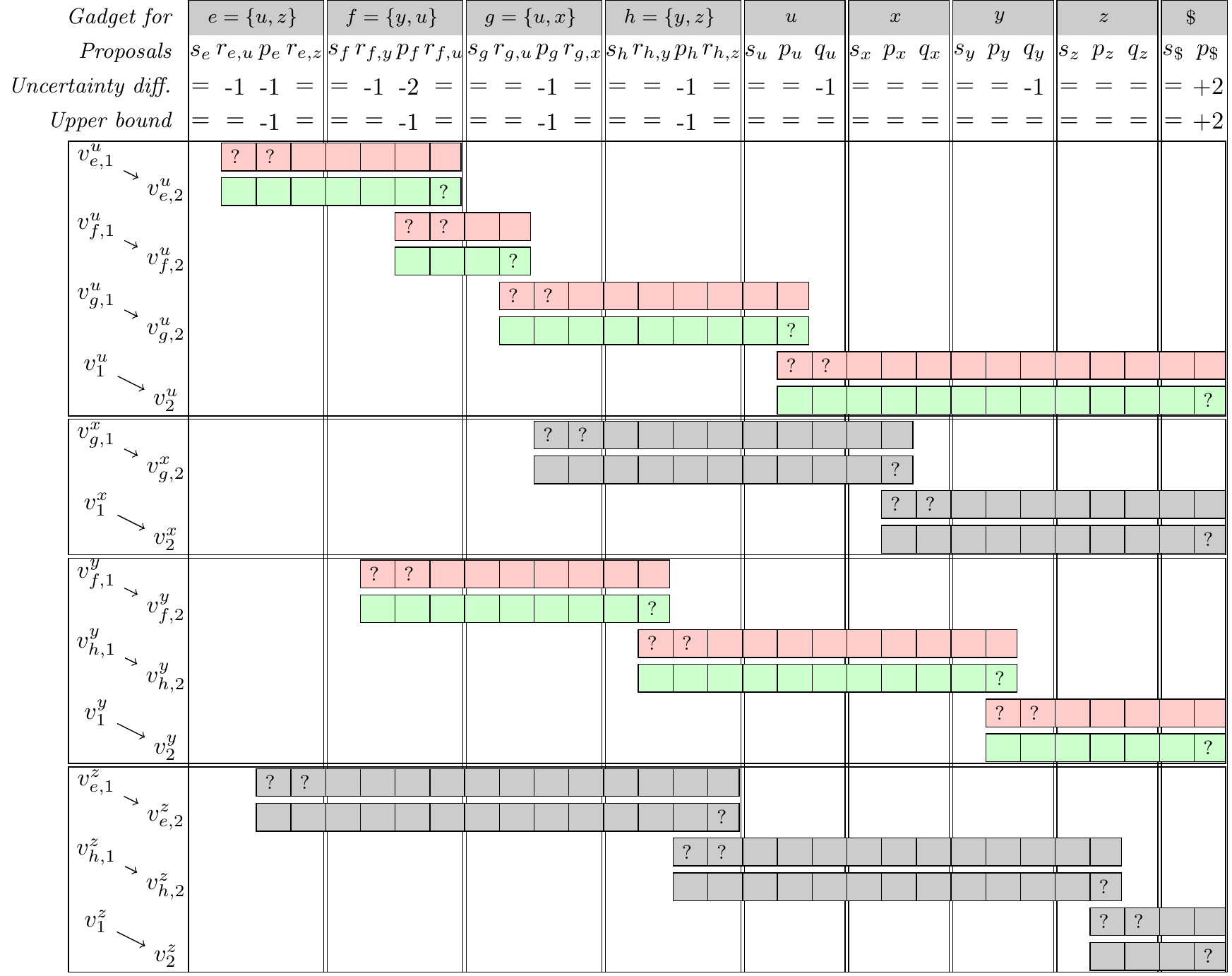}
    \caption{An example of the NP-hardness reduction for CLD in the one-dimensional proposal space (Theorem \ref{thm:1D-CLD_is_hard}), from the instance of Vertex Cover on a four-vertex graph depicted on top. Only coding agents are depicted in the table, which describes the instance constructed by the reduction. A positive opinion for each agent on each proposal is marked with a colored or gray box, while a negative opinion is left blank (due to the restriction on the one-dimensional domain we are considering, colored and gray boxes necessarily form intervals for each agent). Uncertain positions are marked with a question mark. A possible solution (corresponding to the size-$2$ vertex cover $\{u,y\}$ of the original graph) corresponds to having the red (top rows) agents delegate to the corresponding green ones (bottom rows). The number of positive or negative opinions remains unchanged, and there is a net decrease or increase of uncertainty in each column computed in the \emph{Uncertainty diff.} row (where `=' corresponds to a difference of $0$). It can be checked that each difference remains below the upper-bound for each column, as set by balancing agents. The crux of the reduction is that, for each proposal $p_e$, some agent $v^u_{e,1}$ must delegate to $v^u_{e,2}$ to decrease the uncertainty (which corresponds to picking an endpoint of the edge in the vertex cover), but this triggers a chain of delegations resulting in a cost of $1$ in the final proposal; thus, bounding the size of the cover.}
    \label{fig:cldnpfig}
\end{figure*}

\end{document}